%% file: main.tex
\documentclass{article}

\usepackage{graphicx}

\usepackage{amsthm,amsmath,amssymb}
\usepackage{mathtools}

\usepackage{float}
\usepackage{environ}
\usepackage{tabularx}
\usepackage{hyperref}
\usepackage[noabbrev,nameinlink,capitalize]{cleveref}
\usepackage{thmtools}
\usepackage{enumerate} 

\newtheorem{lemma}{Lemma}
\newtheorem{theorem}{Theorem}
\newtheorem{proposition}{Proposition}
\newtheorem{corollary}{Corollary}
\newtheorem{question}{Question}

\hypersetup{
  pdftitle = {Enumerating minimal dominating sets and variants in chordal bipartite graphs},
  pdfauthor = {Emanuel Castelo, Oscar Defrain, Guilherme C. M. Gomes},
  colorlinks = true,
  linkcolor = black!30!red,
  citecolor = black!30!green
}

\input{preamble}

\renewenvironment{abstract}
{\small\vspace{-1em}
\begin{center}
\bfseries\abstractname\vspace{-.5em}\vspace{0pt}
\end{center}
\list{}{
\setlength{\leftmargin}{0.6in}%
\setlength{\rightmargin}{\leftmargin}}%
\item\relax}
{\endlist}

\newcommand{\myparagraph}[1]{\medskip\noindent\textbf{#1}} 
\def\figscale{1.1} 

\usepackage{ifthen}
\newboolean{long-version}
\setboolean{long-version}{true}
\newcommand{\iflongelse}[2]{\ifthenelse{\boolean{long-version}}{#1}{{#2}}}

\usepackage{authblk} 

\title{%
Enumerating minimal dominating sets and variants in chordal bipartite graphs%
\thanks{The second author is supported by the ANR project PARADUAL (ANR-24-CE48-0610-01). The last author is supported by the European Union, project PACKENUM, grant number 101109317. Views and opinions expressed are however those of the author only and do not necessarily reflect those of the European Union. Neither the European Union nor the granting authority can be held responsible for them.}}

\author[1]{Emanuel Castelo}
\author[1]{Oscar Defrain}
\author[2,3]{Guilherme C. M. Gomes}

\affil[1]{Aix-Marseille Université, CNRS, LIS, 13009 Marseille, France}
\affil[2]{LIRMM, Université de Montpellier, CNRS, Montpellier, France}
\affil[3]{Universidade Federal de Minas Gerais, Belo Horizonte, Brazil}

\date{February 2025}

\begin{document}

\maketitle

\begin{abstract}
\input{abstract}

\vskip5pt\noindent{}{\bf Keywords:} algorithmic enumeration, minimal dominating sets, connected dominating sets, total dominating sets, chordal bipartite graphs, sequential method, polynomial delay.
\end{abstract}
\input{content}

\bibliographystyle{alpha}
\bibliography{main}

\end{document}

%% file: abstract.tex
Enumerating minimal dominating sets with polynomial delay in bipartite graphs is a long-standing open problem.
To date, even the subcase of chordal bipartite graphs is open, with the best known algorithm due to Golovach, Heggernes, Kanté, Kratsch, S\ae{}ther, and Villanger running in incremental-polynomial time.
We improve on this result by providing a polynomial delay and space algorithm enumerating minimal dominating sets in chordal bipartite graphs.
Additionally, we show that the total and connected variants admit polynomial and incremental-polynomial delay algorithms, respectively, within the same class.
This provides an alternative proof of a result by Golovach et al.~for total dominating sets, and answers an open question for the connected variant.
Finally, we give evidence that the techniques used in this paper cannot be generalized to bipartite graphs for (total) minimal dominating sets, unless \PNP{}, and show that enumerating minimal connected dominating sets in bipartite graphs is harder than enumerating minimal transversals in general hypergraphs.

%% file: content.tex
\section{Introduction}

Hypergraphs are encountered in countless areas of computer science or discrete mathematics.
Formally, they consist of a pair $\H=(V(\H),E(\H))$ where $V(\H)$ is a set of elements called \emph{vertices}, and $E(\H)$ is a family of subsets of $V(\H)$ called hyperedges (or simply edges).
Hypergraphs are known as \emph{graphs} when their edges are of size precisely two.
A \emph{transversal} in a hypergraph $\H$ is a set $T\subseteq V(\H)$ which intersects every edge of $\H$.
It is called \emph{minimal} if it is inclusion-wise minimal, i.e., if \textit{none} of its proper subsets is a transversal.
The problem of enumerating the minimal transversals of a hypergraph, usually denoted by \transenum{}, is one of the most important problems in algorithmic enumeration, for it has far-reaching implications in various areas of computer science such as logic, database theory, data profiling, or
artificial intelligence~\cite{kavvadias1993horn,eiter1995identifying,gunopulos1997data,blasius2022efficiently}; we refer the reader to~\cite{eiter1995identifying} and the surveys~\cite{eiter2008computational,gainer2017minimal} for a comprehensive overview on the fundamental and practical aspects of this problem.

In a typical enumeration problem such as \transenum{}, a first observation to be made is that the number of solutions can be exponential in both $n$, the number of vertices, and $m$, the number of edges.
Thus, asking for running times which are polynomial in $n+m$ is meaningless.
Rather, we are interested in algorithms performing in output-polynomial time, i.e., that run in a time which is polynomial in both the input size, and the output size.
More strict notions of tractability can be required, such as incremental-polynomial times, and polynomial delay.
We say that an algorithm runs in \emph{incremental-polynomial} time if it produces the $i^\text{th}$ solution in a time which is polynomial in the input size plus $i$.
An algorithm is said to run with \emph{polynomial delay} if the times before the first output, after the last output, and in between two consecutive outputs are bounded by a polynomial in the input size only.
To date, the best-known algorithm solving \transenum{} is due to Fredman and Khachiyan~\cite{fredman1996complexity} and runs in quasi-polynomial incremental time. 
Since then, output-polynomial time algorithms were obtained for several important classes, including those generalizing that of bounded degree~\cite{eiter2003new}, or bounded edge size~\cite{khachiyan2007conformality}.

Dominating sets are among the most studied objects in graphs.
A \emph{dominating set} in a graph $G$ is a set $D\subseteq V(G)$ such that every vertex is either in $D$, or adjacent to a vertex of $D$.
Again, it is called \emph{minimal} if it is inclusion-wise minimal.
The problem of enumerating minimal dominating sets in graphs is denoted by \domenum{}.
Note that being a dominating set of $G$ can be rephrased as being a transversal of each closed neighborhood of $G$, so, \domenum{} is a special case of \transenum{} with a linear number of hyperedges. 

The problem \domenum{} has gained lots of interest since it has been shown, perhaps surprisingly, to be polynomially equivalent to \transenum{}, even when restricted to cobipartite graphs~\cite{kante2014split}.
In particular, this work has initiated a fruitful line of research of characterizing which graph classes admit output-polynomial algorithms, in the hope of better understanding the underlying structure that makes \transenum{} tractable.
For example, polynomial-delay algorithms were obtained for various graph classes including chordal graphs \cite{kante2015chordal} (with linear delay on some subclasses \cite{kante2014split,defrain2019neighborhood}), graphs of bounded degeneracy~\cite{eiter2003new}, line graphs \cite{kante2015line}, graphs of bounded clique-width \cite{courcelle2009linear}, or graph classes related to permutation graphs and their generalizations~\cite{kante2012neighbourhood, golovach2018lmimwidth, bonamy2020comp}.
Incremental-polynomial time algorithms were obtained for graphs of bounded conformality~\cite{khachiyan2007conformality}, $C_{\leq 6}$-free graphs~\cite{golovach2015flipping}, chordal bi\-partite graphs \cite{golovach2016enumerating}, 
$(C_{4},C_{5},\text{claw})$-free graphs \cite{kante2012neighbourhood}, 
or graphs arising in geometrical contexts~\cite{elbassioni2019global,golovach2018lmimwidth} including unit-disk and unit-square graphs.
Output-polynomial time algorithms are known for $\log n$-degenerate graphs~\cite{eiter2003new}, and graphs of bounded clique number \cite{bonamy2020kt}.

In~\cite{kante2014split}, the authors also include two natural variants of domination in their study: total and connected dominating sets.
A \emph{total dominating set} is a set $D$ such that every vertex of $G$ has a neighbor in $D$.
In other words, elements of $D$ cannot ``dominate'' themselves and need a neighbor in the set.
A \emph{connected dominating set} is a dominating set that induces a connected subgraph.
Note that any connected dominating set is a total dominating set, as long as it has more than one element, making the intersection of these two notions irrelevant to study.
Perhaps not surprisingly, the authors in \cite{kante2014split} show that the problems of enumerating minimal total and connected dominating sets, denoted \tdomenum{} and \cdomenum{}, respectively, are equivalent to \transenum{} when restricted to split graphs.
These equivalences on a relatively restricted graph class suggest that 
these variants of \domenum{}
are challenging, and may explain why they have received less attention in the literature.

Regarding total dominating sets, a first observation is that they are transversals of open neighborhoods. This implies an incremental quasi-polynomial time algorithm for the problem using the results of Fredman and Khachiyan~\cite{fredman1996complexity}.
However, few output-polynomial time algorithms are known for specific instances of \tdomenum{}.
Among them, we cite the class of chordal bipartite graphs, which admits a polynomial-delay algorithm for the problem \cite{golovach2016enumerating}. 

Concerning connected dominating sets, it is not known whether the problem reduces to \transenum{}. 
It has in fact been conjectured that this is not the case in~\cite{lorentz2015open}.
Rather, it has been shown in \cite{kante2014split} that the problem reduces to the enumeration of minimal transversals of an implicit hypergraph given by the minimal separators of the graph.
However, the number of such separators may be exponential, and are intractable to compute~\cite{brosse2024hardness}.
The only work known to us dealing with connected dominating sets enumeration, apart from~\cite{kante2014split}, falls under the input-sensitive approach~\cite{golovach2016connected,golovach2020connected,abu2022connected} where one aims at small exponential times.
Still, we note that some essential questions concerning output-sensitive enumeration were addressed in~\cite{abu2022connected} such as the \NP-hardness of the associated extension problem, i.e., deciding if a given set $U \subseteq V(G)$ is contained in some minimal connected dominating set.

Note that, in the literature, particular attention has been put into characterizing the status of these problems on bipartite, split, and cobipartite graphs.
This is because such instances contain the incidence graph of any hypergraph, which is fundamentally challenging and makes interesting links with \transenum{}.
However, and in contrast to the split and cobipartite cases, the case of bipartite graphs has been particularly challenging with the question of the (in)tractability of these problems raised in various works~\cite{kante2014encyclopedia,golovach2016enumerating}. 
Only recently has an output-polynomial time algorithm been provided for \domenum{} in a generalization of bipartite graphs~\cite{bonamy2020kt}.
Still the problem is not known to admit a polynomial-delay algorithm.
Concerning \tdomenum{} and \cdomenum{}, it is also open, to the best of our knowledge, whether they can be solved with polynomial delay, or even in output-polynomial time, in bipartite graphs.
Among natural subclasses of bipartite graphs, the case of chordal bipartite graphs has been studied as it admits nice structural properties analogous to chordal graphs, such as perfect elimination orderings and linearly many minimal separators.
We note, however, that even in this particular class, the existence of a polynomial-delay algorithm for \domenum{} seems open, and has been left open in~\cite{golovach2016enumerating}.

The following two natural questions arise from the literature.

\begin{question}\label{qu:mds}
    Can the minimal dominating sets of a chordal bipartite graph be enumerated with polynomial delay?
\end{question}

\begin{question}\label{qu:mcds}
    Can the minimal connected dominating sets of a chordal bipartite graph be enumerated in output-polynomial time?
\end{question}

Only the case of minimal total dominating sets has been settled for chordal bipartite graphs, with the following result.

\begin{restatable}[{\cite{golovach2016enumerating}}]{theorem}{thmtds}
\label{thm:mtds-ch-bip}
    The minimal total dominating sets of a chordal bipartite graph can be enumerated with polynomial delay.
\end{restatable}

\myparagraph{Our results and techniques.} This paper is devoted to providing a positive answer to 
\hyperref[prop:mds-closed-neighborhoods]{Questions \ref*{qu:mds}} and~\ref{qu:mcds},
almost completing the picture for chordal bipartite graphs. 
Collaterally, we provide an alternative proof of \autoref{thm:mtds-ch-bip} involving small modifications of our algorithm for minimal dominating sets; we note that the obtained algorithm has similar time bounds as the one in \cite{golovach2016enumerating}.

Our algorithms for \domenum{} (\autoref{thm:mds-ch-bip}) and \tdomenum{} (\autoref{thm:mtds-ch-bip}) are based on the \emph{ordered generation} framework (also known as \emph{sequential method} in various works) which has been introduced in \cite{eiter2003new} as a generalization of the algorithm in \cite{tsukiyama1977new} (see also \cite{johnson1988generating} for a concise description). 
This approach has since then proved to be fruitful in various contexts related to domination and minimal transversals~\cite{conte2019irredundant,bonamy2020kt,defrain2021translating}.
Roughly, the idea is to decompose the instance with respect to an ordering of its elements, and, given the partial solutions of the instance induced by the $i$ first elements, to extend them to the next element of the instance. 
To achieve polynomial delay, a parent-child relation on partial solutions is defined. 
Then, the resulting tree is traversed in a DFS manner, and only the leaves (which correspond to actual solutions) are output.
Overall, the technique allows to reduce the general enumeration to the enumeration of the children of a partial solution, and is particularly suitable for instances admitting good elimination orderings.
In our case, we use the weak-elimination ordering of chordal bipartite graphs (see~\autoref{sec:prelim}) to solve children generation in polynomial time, yielding a polynomial-delay algorithm for \domenum{} and \tdomenum{} in that class.

Concerning \cdomenum{}, we show that the minimal separators of a chordal bipartite graph define a hypergraph of bounded conformality (\autoref{lemma:conf}), a parameter introduced  by Berge~\cite{berge1984hypergraphs} and successfully used in the design of tractable algorithms for domination~\cite{kante2012neighbourhood} and minimal transversals enumeration \cite{khachiyan2007conformality}.
This, combined with the fact that chordal bipartite graphs have polynomially many minimal separators, which can be listed in polynomial time (see \autoref{sec:prelim}), allows us to obtain an incremental-polynomial time algorithm for \cdomenum{} in that class (\autoref{thm:mcds-incp}).

Additionally, we give evidence that the sequential method cannot be directly used to generalize our results to bipartite graphs. 
More specifically, we show that the problem of generating the children of a given partial solution is intractable for \domenum{} (\autoref{lem:hardness-sequential-bipartite}) and \tdomenum{} \iflongelse{ (\autoref{lem:hardness-sequential-bipartite-mtds})}{(\autoref{sec:sequential-limit})}.
For \cdomenum{} we make a stronger statement (\autoref{thm:mcds-hardness}): the problem is at least as hard as \transenum{}, which was not known to the best of our knowledge, despite the reduction being straightforward. 

\myparagraph{Organization of the paper.}
We introduce preliminary notions and properties dealing with chordal bipartite graphs and minimal separators in \autoref{sec:prelim}.
We describe the sequential approach in \autoref{sec:sequential-approach}, and provide our algorithms for \domenum{} and \tdomenum{} in \autoref{sec:mds-and-mtds}.
The arguments for an incremental-polynomial time algorithm for \cdomenum{} in chordal bipartite graphs are given in \autoref{sec:mcds}.
Limitations of the sequential method in bipartite graphs are discussed in \autoref{sec:sequential-limit}, and the reduction from \transenum{} to \cdomenum{} in bipartite graphs is given in \autoref{sec:mcds-hardness}.
We discuss open questions in \autoref{sec:conclusion}.

\section{Preliminary notions}\label{sec:prelim}

We assume familiarity with standard terminology from graph and hypergraph theory, and refer to~\cite{diestel2012graph,berge1984hypergraphs} for notations not included below.

\myparagraph{Graphs.} 
All graphs considered in this paper are simple and finite.
Given $X \subseteq V(G)$, $G[X] \coloneq (X, \{uv \in E(G) \mid u,v \in X\})$ is the subgraph of $G$ \emph{induced} by $X$.
Given $v \in V(G)$, we denote by $N(v)$ its open neighborhood, by $N[v]$ its closed neighborhood, and define its \emph{distance-2 neighborhood} as $N^2(v)\coloneq N[N[v]]\setminus N[v]$.
The first two notions are naturally extended to subsets by $N[S]\coloneq \bigcup_{v\in S} N[S]$ and $N(S)\coloneq (\bigcup_{v\in S} N[S])\setminus S$.
An \emph{independent set} is a set of pairwise non-adjacent vertices.
A graph is \emph{bipartite} if its vertex set $V(G)$ can be partitioned into two independent sets $X,Y$.
It is called \emph{bipartite chain} if, in addition, there exist orderings $x_1,\dots,x_p$ and $y_1,\dots,y_q$ of $X$ and $Y$, respectively, such that $N(x_i)\subseteq N(x_j)$ for all $1\leq i\leq j\leq p$, and $N(y_i)\supseteq N(y_j)$ for all $1\leq i\leq j\leq q$.
Note that bipartite chain graphs are closed under taking induced subgraphs.

\myparagraph{Hypergraphs.}
Let $\H$ be a hypergraph.
Given a vertex $v\in V(\H)$, the set of edges that contain $v$, also called edges \emph{incident} to $v$, is denoted by $\inc(v)$.
The hypergraph \emph{induced} by a set $X\subseteq V(\H)$ of vertices is defined as $\H[X]=(X, \{E\in E(\H): E\subseteq X\})$.
A \emph{transversal} of $\H$ is a subset $T\subseteq V(\H)$ such that $E\cap T\neq \emptyset$ for all $E\in E(\H)$, i.e., every edge is incident to (or \emph{hit} by) a vertex of $T$.
It is \emph{minimal} if it is minimal by inclusion, or equivalently, if $T\setminus\{v\}$ is not a transversal for any $v\in T$.
The set of minimal transversals of $\H$ is denoted by $Tr(\H)$.
The \emph{private edges} of an element $v$ in a transversal $T$ are those edges of $\H$ which are only hit by $v$ in $T$; we denote the set of such edges by $\priv(T,v)$.
It is well-known (see e.g.,~\cite{berge1984hypergraphs}) that a set $T\subseteq V(\H)$ is a minimal transversal of $\H$ if and only if $T$ is a transversal and every $v$ in $T$ has a private edge. We formally state the problem of enumerating minimal transversals of a hypergraph as follows.

\begin{problemgen}
  \problemtitle{\textsc{Minimal Transversals Enumeration (\transenum{})}}
  \probleminput{A hypergraph $\H$.}
  \problemquestion{The set $Tr(\H)$.}
\end{problemgen}

Note that in the problem above, $\H$ can be considered inclusion-wise minimal (or \emph{Sperner}) as this does not change its set of minimal transversals, and removing non-minimal hyperedges can be performed in polynomial time \cite{eiter2008computational}.

\myparagraph{Domination.}
We recall the different notions of domination given in the introduction.
A \emph{dominating set} (or DS for short) in a graph $G$ is a set $D\subseteq V(G)$ such that $V(G)=N[D]$.
It is called \emph{connected} if the graph $G[D]$ is connected.
A \emph{total dominating set} in $G$ is a set $D\subseteq V(G)$ such that $V(G)=\bigcup_{v\in D} N(v)$.
A (resp.~connected, total) dominating set is said to be \emph{minimal} if it is inclusion-wise minimal.
We denote by $\mds(G)$ (resp.~$\mcds(G)$, $\mtds(G)$) the families of minimal such sets in a graph $G$.
\iflongelse{}{The problems of enumerating these families are our main interest in this paper.}

\iflongelse{
In this paper, we are interested in the following enumeration problems.

\begin{problemgen}
  \problemtitle{\textsc{Minimal DS Enumeration (\domenum{})}}
  \probleminput{A graph \(G\).}
  \problemquestion{The set $\mds(G)$.}
\end{problemgen}

\begin{problemgen}
  \problemtitle{\textsc{Minimal Total DS Enumeration (\tdomenum{})}}
  \probleminput{A graph \(G\).}
  \problemquestion{The set $\mtds(G)$.}
\end{problemgen}

\begin{problemgen}
  \problemtitle{\textsc{Minimal Connected DS Enumeration (\cdomenum{})}}
  \probleminput{A graph \(G\).}
  \problemquestion{The set $\mcds(G)$.}
\end{problemgen}}
{}

We recall that the precise complexity status of each of the problems above is widely open, the first one admitting a quasi-polynomial time algorithm~\cite{fredman1996complexity,kante2014split}, while the third is conjectured to be intractable~\cite{lorentz2015open}.
Their study has been initiated in~\cite{kante2014split}, where it is noted that each of the problems is a particular instance of \transenum{}: using the family of (closed and open) neighborhoods for the first two, and a more involved hypergraph for the third one, as formalized in \hyperref[prop:mds-closed-neighborhoods]{Propositions \ref*{prop:mds-closed-neighborhoods}}--\ref{prop:mcds-separators}.
See also~\cite{berge1984hypergraphs}.

\begin{proposition}\label{prop:mds-closed-neighborhoods}
    For every graph $G$, $\mds(G)=Tr(\{N[v] : v\in V(G)\})$.
\end{proposition}

\begin{proposition}\label{prop:mds-open-neighborhoods}
    For every graph $G$, $\mtds(G)=Tr(\{N(v) : v\in V(G)\})$.
\end{proposition}

\myparagraph{Minimal separators.}
A \emph{separator} in a graph $G$ is a subset $S\subseteq V(G)$ such that $G-S$ is disconnected. In the following, let $\S(G)\coloneq  \{S\subseteq V : G-S~\text{is not connected and}~\forall S' \subset S,  G - S'~\text{is connected}\}$ denote the set of (inclusion-wise) minimal separators of $G$.
Then we have the following.

\begin{proposition}[{\cite{kante2014split}}]\label{prop:mcds-separators}
    For every graph $G$, $\mcds(G)=Tr(\S(G))$.
\end{proposition}

Minimal separators are not to be confused with the slightly different notion of \emph{minimal $a$-$b$ separators} which are minimal sets $S$ such that $a$ and $b$ lie in different components of $G-S$.
Indeed, while every minimal separator is a minimal $a$-$b$ separator, the converse is not always true: minimal $a$-$b$ separators may contain another $c$-$d$ separator as a subset; see~\cite{golumbic2004algorithmic} for an early mention of this observation, and \cite{brosse2024hardness} for a thorough discussion on this matter.
A simple example is the cycle on four vertices with a pendant edge \cite{brosse2024hardness}.
Moreover, while enumerating $a$-$b$ minimal separators can be done with polynomial delay~\cite{shen1997efficient,kloks1998listing,berry2000generating}, enumerating minimal separators is a hard task~\cite{brosse2024hardness}.
However, and as is noted in \cite{kante2014split}, \cdomenum{} reduces to \transenum{} for graph classes admitting polynomially many $a$-$b$ separators.\footnote{%
    It should be noted that the statement \cite[Corollary 33]{kante2014split} presents an inaccuracy where by \emph{minimal separators} it should be understood \emph{minimal $a$-$b$ separators}. Indeed, listing minimal separators is hard~\cite{brosse2024hardness} while enumerating minimal $a$-$b$ separators is tractable~\cite{kloks1998listing}. 
    The present inaccuracy comes from an ambiguity and common confusion in the literature around this terminology \cite{shen1997efficient,kloks2011feedback}.}
Let $\G$ be the hypergraph where each hyperedge is an $a$-$b$ minimal separator of $G$; using the above-cited algorithms for enumerating $a$-$b$ minimal separators, we can construct $\G$ in polynomial time on $|\G| + |V(G)| + |E(G)|$; moreover, note that $\S(G)=\Min_\subseteq \G$.
Indeed, recall that $Tr(\H)=Tr(\G)$ for any pair of hypergraphs $\H,\G$ such that for any $E\in E(\G)$ there exists $E'\in E(\H)$ with $E'\subseteq E$. In other words, $Tr(\H)=Tr(\Min_\subseteq \H)$ for any hypergraph $\H$~\cite{berge1984hypergraphs}.

\myparagraph{Chordal bipartite graphs.}
A bipartite graph is \emph{chordal bipartite} if it does not contain induced cycles of length greater than four~\cite{golumbic2004algorithmic}.
Note that chordal bipartite graphs are \emph{not} chordal, as they may contain the induced cycle on four vertices.
This terminology should be regarded as an analogue to bipartite graphs of chordality, as bicliques are to cliques.

A vertex $v$ in a graph $G$ is called \emph{weak-simplicial} if $N(v)$ is an independent set, and for any two $x,y\in N(v)$, we have $N(x)\subseteq N(y)$ or $N(y)\subseteq N(x)$; two such vertices $x,y$ are said to be \emph{comparable}.
Note that this is precisely to say that the graph $G[N(v)\cup N^2(v)]$ induced by the vertices at distance at most $2$ from $v$ is a bipartite chain; see~\cite[Lemma~3]{kurita2019efficient}.
A \emph{weak-simplicial elimination ordering} of $G$ is an ordering $v_1,\dots,v_n$ of its vertices such that, for all $1\leq i\leq n$, $v_i$ is weak-simplicial in $G_i = G\left[\{v_1, \dots, v_i\}\right]$. 
The following characterization is due to \cite{kurita2019efficient} and may be regarded as a simplification of an earlier characterization and recognition algorithm by Uehara~\cite{uehara2002linear}. 

\begin{proposition}[{\cite{uehara2002linear,kurita2019efficient}}]\label{prop:weak-seo}
    A graph is chordal bipartite if and only if it admits a weak-simplicial elimination ordering, which can be computed in polynomial time.
\end{proposition}

Another characterization of chordal bipartite graphs based on separators is given in \cite{golumbic1978perfect}, where it is proved that a bipartite graph is chordal bipartite if and only if every minimal \emph{edge} separator induces a complete bipartite subgraph.
In~\cite{kloks2011feedback}, the following analogue of the forward implication is given for minimal $a$-$b$ separators (hence for minimal separators). 

\begin{proposition}[\cite{kloks2011feedback}]\label{prop:chbip-vertex-sep}
    If $G$ is chordal bipartite and $S$ is an $a$-$b$ minimal separator then $G[S]$ is complete bipartite.
\end{proposition}

In the same paper, the authors show that chordal bipartite graphs have at most $O(n+m)$ minimal $a$-$b$ separators \cite{kloks2011feedback}.
Since they can be listed with polynomial delay~\cite{kloks1998listing}, we derive a polynomial time procedure listing $a$-$b$ minimal separators and filtering out the non-minimal ones.
We formalize this in the next statement that will be of use in \autoref{sec:mcds}.

\begin{proposition}\label{prop:chbip-sep-enum}
    Let $G$ be an $n$-vertex chordal bipartite graph. Then its set $\S(G)$ of minimal separators can be enumerated in polynomial time in $n$.
\end{proposition}

Finally, we will need two extra properties of minimal separators in general and in chordal bipartite graphs.
Let \(S\) be a separator of \(G\) and $C$ be a component of $G-S$. We say that \(C\) is \emph{close} to \(S\) if every vertex in \(S\) has at least one neighbor in \(C\).
It is folklore that minimal $a$-$b$ separators $S$ are characterized by the existence of two close components in $G-S$; see~\cite{shen1997efficient,golumbic2004algorithmic,kloks2011feedback} or \cite{kloks1994finding} for a proof.
We derive one implication for minimal separators since, in particular, they are $a$-$b$ minimal separators for some specific pairs $a,b$ of vertices. 
Note however that the converse is not true.

\begin{proposition}\label{prop:close}
    Let $G$ be a graph and $S$ be a minimal separator of $G$.
    Then there exist two components $C_1,C_2$ in $G-S$ that are close to $S$.
\end{proposition}

Finally, the following will be of use in Section~\ref{sec:mcds}. We state it for minimal separators but this more generally holds for minimal $a$-$b$ separators~\cite{kloks2011feedback}.

\begin{proposition}[{\cite{kloks2011feedback}}]\label{prop:close-neighbor}
Let $G$ be chordal bipartite with bipartition $(A,B)$. Let $S$ be a minimal separator of $G$ and $C$ be a component of $G-S$ that is close to $S$. 
If $S\cap A\neq\emptyset$ then there exists a vertex $x$ in $C$ with
$N(x) \cap S = S \cap A$. 
\end{proposition}

\section{The sequential method}\label{sec:sequential-approach}

We describe the \emph{sequential method}.
As by \hyperref[prop:mds-closed-neighborhoods]{Propositions \ref*{prop:mds-closed-neighborhoods}}--\ref{prop:mcds-separators}, 
the problems \domenum{}, \tdomenum{}, and \cdomenum{} can be seen as particular instances of \transenum{} we choose to follow the presentation of~\cite{bartier2024hypergraph}, and refer the reader to~\cite{eiter2003new,elbassioni2008some,bonamy2020kt} for more details on this approach.

Let \(\H \subseteq 2^{V}\) be a hypergraph on vertex set \(V=\{v_{1}, \dots, v_{n}\}\). 
For \(1\leq i\leq n\) we denote by \(V_{i} \coloneq  \{v_{1}, \dots, v_{i}\}\) the set of its first \(i\) vertices, by \(\H_{i} \coloneq  \H[V_{i}]\) the hypergraph induced by \(V_{i}\), and set \(V_0\coloneq \H_{0} \coloneq  Tr(\H_{0}) \coloneq  \emptyset\). 
In addition, for \(S \subseteq V_i\) and \(v \in S\), we note \(\priv_{i}(S, v)\) the set of private edges of \(v\) with respect to $S$ in \(\H_{i}\).
Given a minimal transversal \(T \in Tr(\H_{i+1})\) for $i<n$, we denote by \(\parent(T, i + 1)\) the \emph{parent of \(T\) with respect to \(i + 1\)} obtained after repeatedly removing from $T$ the vertex of smallest index \(v\) such that \(\priv_{i}(T, v) = \emptyset\). 
By definition, every \(T\) has a unique parent. 
Lastly, we denote by \(\children(T^{\star}, i)\), for \(T^{\star} \in Tr(\H_{i})\), the family of \(T \in Tr(\H_{i+1})\) such that \(T^{\star} = \parent(T, i + 1)\). 

The next two lemmas are central to the description of the algorithm: they define a tree structure where nodes are pairs $(T,i)$ with $T\in Tr(\H_i)$ and that will be traversed in a DFS manner in order to produce every leaf $T\in Tr(\H)=Tr(\H_n)$ as a solution.
We refer to \cite{bartier2024hypergraph} for the proofs of these statements, which follow from the definition of the $\parent$ relation. 

\begin{lemma}
    Let \(0 \leq i \leq n - 1\) and \(T \in Tr(\H_{i+1})\), then \(\parent(T, i + 1) \in Tr(\H_{i})\).
\end{lemma}

\begin{lemma}\label{lem:at-least-one-child}
    Let \(0 \leq i \leq n - 1\) and \(T^{\star} \in Tr(\H_{i})\), then either:
    \begin{itemize}
        \item \(T^{\star} \in Tr(\H_{i+1})\) and consequently \(\parent(T^{\star}, i + 1) = T^{\star}\); or
        \item \(T^{\star} \cup \{v_{i+1}\} \in Tr(\H_{i+1})\) and \(\parent(T^{\star} \cup \{v_{i+1}\}, i + 1) = T^{\star}\).
    \end{itemize} 
\end{lemma}

Finally, the whole framework culminates in the following statement that reduces \transenum{} to the enumeration of the children of a given pair $(T,i)$ with $T\in Tr(\H_i)$, $i<n$, while preserving polynomial delay and space.

\begin{theorem}
    Let \(f, s: \N^{2} \rightarrow \Z_{+}\) be two functions. If there is an algorithm that, given a hypergraph \(\H\) on vertex set \(V = \{v_{1}, \dots, v_{n}\}\) and $m$ edges, an integer \(i \in \intv{n-1}\), and \(T^{\star} \in Tr(\H_{i})\), enumerates \(\children(T^{\star}, i)\) with delay \(f(n, m)\) and using \(s(n, m)\) space, then there is an algorithm that enumerates \(Tr(\H)\) with \(O(n f(n, m))\) delay and \(O(n s(n, m))\) space. 
\end{theorem}

In the following, given $T^{\star}$ and $i$ as above, let 
\[
    \Delta_{i+1}\coloneq \{E\in \H_{i+1} : E\cap T^{\star}=\emptyset\}\subseteq \inc_{i+1}(v_{i+1})
\]
be the family of hyperedges of $\H_{i+1}$ that are not hit by $T^\star$.
The following shows that children may be found by considering minimal extensions to $\H_{i+1}$ of minimal transversals of $\H_{i}$.

\begin{lemma}\label{lem:extensions}
    If $T\in \children(T^\star,i)$ then $T=T^\star\cup X$ for some $X\in Tr(\Delta_{i+1})$.
\end{lemma}

\begin{proof}
    Let $T\in \children(T^\star,i)$.
    Then $T^\star$ is obtained from $T$ by greedily removing elements that have no private edge in $\H_i$.
    Let $X$ be the set of removed elements.
    Note that since $T^\star$ is transversal of $\H_i$, $X$ may only have private edges among those $E\in \H_{i+1}$ that are not intersected by $T^\star$, i.e., among $\Delta_{i+1}$.
    Then, either $\Delta_{i+1}=\emptyset$ in which case $T=T^\star$ by \autoref{lem:at-least-one-child}, or $X$ is non empty and intersects each edge in $\Delta_{i+1}$.
    By minimality of $T$, we derive that $X$ is minimal with this property, as desired.
\end{proof}

As a corollary, in order to enumerate $\children(T^{\star}, i)$ for some $i\in \intv{n-1}$ and $T^{\star}\in Tr(\H_i)$, one can simply enumerate all the minimal transversals of $\Delta_{i+1}$ and filter out those that do not provide children of $T^\star$.
In the following, we call such sets \emph{minimal extensions} of $(T^{\star}, i)$.
Note, however that the number of minimal transversal of $\Delta_{i+1}$ may be greater (typically exponential) than the number of actual children, hence that this approach fails to efficiently produce children in general.
Moreover, this approach is intractable in general, with instances admitting non-trivial children if and only if NP-complete problems admit a positive answer; see~\cite{bartier2024hypergraph} or \autoref{sec:sequential-limit}.
We nevertheless have the following which says that the approach is still tractable if we can bound the number of minimal extension by a polynomial in the input size, and which is in essence what is done in~\cite{eiter2003new}.

\begin{corollary}\label{cor:all-extensions}
    If there is an algorithm that, given \(T^{\star} \in Tr(\H_{i})\), $i<n$, enumerates all $X\in Tr(\Delta_{i+1})$ in time and space which are polynomial in $n+m$, 
    then there is one that lists \(Tr(\H)\) with polynomial delay and space.
\end{corollary}

Note that, if $\H$ is $d$-degenerate and $v_1,\dots,v_n$ is an optimal degeneracy ordering, then minimal extensions have size at most $d$.
In~\cite{eiter2003new}, the authors use this fact to provide a polynomial delay algorithm that lists minimal transversals in $O(1)$-degenerate hypergraphs relying on \autoref{cor:all-extensions}.
See~\cite{eiter2003new,bartier2024hypergraph} for the definition of these notions.
In this paper, we combine this corollary together with the elimination ordering of chordal bipartite graphs provided by~\autoref{prop:weak-seo} to show that \domenum{} and \tdomenum{} can be solved with polynomial delay as well in this class. 

\section{Minimal extension for (total) dominating sets}\label{sec:mds-and-mtds}

In this section, we apply the sequential method to \transenum{} restricted to hypergraphs of closed neighborhoods of chordal bipartite graphs to show that \autoref{qu:mds} admits a positive answer.
More specifically, we use the weak-simplicial elimination ordering of chordal bipartite graphs to restrict the instance of minimal extensions enumeration to a variant of domination in bipartite chain graphs, which is tractable.
By \autoref{cor:all-extensions}, this allows us to obtain \autoref{thm:mds-ch-bip}.

In the following, let $G$ be a chordal bipartite graph on $n$ vertices and $v_1,\dots,v_n$  
be a weak-simplicial ordering of its vertices as provided by \autoref{prop:weak-seo}.
Let $\H$ be its hypergraph of closed neighborhoods, $ i \in \intv{n-1}$ be an integer and $T$ be a minimal transversal of $\H_i$ as in the statement of \autoref{cor:all-extensions}.
We will show how to enumerate the minimal transversals of $\Delta_{i+1}$ in polynomial time and space, to derive the desired result. 

First note that since the edges of $\Delta_{i+1}$ belong to $\H_{i+1}$ and are incident to $v_{i+1}$, they may be equal to $N[v_{i+1}]$ if $N[v_{i+1}]\subseteq V_{i+1}$, or to $N[u]$ for some neighbor $u$ of $v_{i+1}$ if $N[u]\subseteq V_{i+1}$.
Note that in the later case, the set $N[u]$ may contain vertices at distance 2 from $v_{i+1}$.
From now on, let $B$ be the set of vertices $u\in N(v_{i+1})$ such that $N[u]\in \Delta_{i+1}$, and let
\[
    R\coloneq \left(\,\bigcup\{E: E \in \Delta_{i+1}\}\right) \setminus B
\]
be the remaining vertices that lie in some edge of $\Delta_{i+1}$.
Consequently $B\subseteq N(v_{i+1})$ and $R\subseteq N(v_{i+1})\cup N^2(v_{i+1})$.
Thus the graph $H=G_i[R\cup B]$ is an induced subgraph of $G_i[N(v_{i+1})\cup N^2(v_{i+1})]$ which, by \autoref{prop:weak-seo}, is a bipartite chain.
Hence $H$ is a bipartite chain.
We denote by $(X,Y)$ the bipartition of $H$ where $B\subseteq X\subseteq N(v_{i+1})$, and by $x_1,\dots,x_p$ and $y_1,\dots,y_q$ the respective orderings of $X$ and $Y$ satisfying $N(x_j)\cap Y\subseteq N(x_\ell)\cap Y$ and $N(y_\ell)\cap X\subseteq N(y_j)\cap X$ for all $j<\ell$; see~\autoref{sec:prelim} and \autoref{fig:delta_i} for an illustration of the situation.

\begin{figure}
    \centering
    \includegraphics[scale=\figscale]{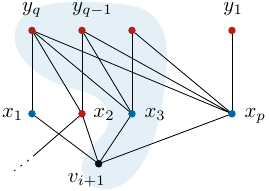}
    \caption{The bipartite chain instance induced by $X$ and $Y$ with $v_{i+1}$ depicted as well. Blue vertices are those whose closed neighborhoods are incident to $v_{i+1}$ \emph{and} included in $V_{i+1}$ as a subset. Red vertices in $X$ have neighbors out of $V_i$.}
    \label{fig:delta_i}
\end{figure}

In what remains we shall call \emph{blue} the vertices in $B$, \emph{red} those in $R$, and characterize minimal transversals of $\Delta_{i+1}$ as particular sets analogous to red-blue dominating sets where red and blue vertices may be selected in order to dominate blue vertices.

\begin{proposition}\label{prop:red-part}
    Let $Z\in Tr(\Delta_{i+1})$. 
    Then 
    \begin{itemize}
        \item $|Z \cap R \cap N(v_{i+1})|\leq 1$; and
        \item $|Z\cap N^2(v_{i+1})|\leq 1$.
    \end{itemize}
\end{proposition}

\begin{proof}
    First, note that if $x$ belongs to $Z \cap R \cap N(v_{i+1})$ then it is incident to at most one edge in $\Delta_{i+1}$ (as $H$ is bipartite), namely $N[v_{i+1}]$, in which case $N[v_{i+1}]$ is a private edge of $x$.
    Hence, no two distinct vertices $x,y$ may lie in $Z \cap R \cap N(v_{i+1})$ by minimality of $Z$, yielding the first inequality.

    Now, recall that $N^2(v_{i+1})\subseteq Y$ hence that one of $N(y)\cap X\subseteq N(y')\cap X$ or $N(y)\cap X\supseteq N(y')\cap X$ holds for any pair $y,y'\in N^2(v_{i+1})$.
    Since edges of $\Delta_{i+1}$ are of the form $N[x]$ for some $x\in X$, no two distinct $y,y'$ can have private edges in $\Delta_{i+1}$ simultaneously.
    Hence, $|Z\cap N^2(v_{i+1})|\leq 1$ as claimed.
\end{proof}

\begin{lemma}\label{lem:extension-characterization}
    Let $Z\in Tr(\Delta_{i+1})$. 
    Then exactly one of the following holds:
    \begin{enumerate}
        \item\label{it:vi} \(Z = \{v_{i+1}\}\); 
        \item\label{it:blue} \(Z \subseteq B\) in which case $Z=B$; 
        \item\label{it:red} \(Z \subseteq R\) in which case \(|Z| \leq 2\); or
        \item\label{it:blue-red} \(Z = \{r\} \cup (B \setminus N(r))\) for some \(r \in N^{2}(v_{i+1})\).
    \end{enumerate}

    
    \begin{proof}
        First, note that all the cases are mutually distinct.
        Hence, we may consider some $Z\in Tr(\Delta_{i+1})$, assume that $Z\neq \{v_{i+1}\}$, and prove the conclusion for each of 
        \hyperref[it:blue]{Items \ref*{it:blue}}--\ref{it:blue-red} 
        Note that each $x\in X$ hits a single edge $N[x]$ in $\Delta_{i+1}$, as $X$ is an independent set.
        Hence, if $Z\subseteq B$ then $Z=B$ concluding \hyperref[it:blue]{Item \ref*{it:blue}}.
        If $Z\subseteq R$ then we have that $|Z| \leq 2$ by \autoref{prop:red-part} deriving \hyperref[it:red]{Item \ref*{it:red}}.
        We are thus left with proving \hyperref[it:blue]{Item \ref*{it:blue-red}}.
        Suppose that \( Z \cap B \neq \emptyset\) and \(Z \cap R \neq \emptyset\).  
        Note that if $x\in X\cap R$, then it may be incident to at most one edge in $\Delta_{i+1}$, namely $N[v_{i+1}]$.
        However, since \(Z \cap B \neq \emptyset\) by assumption of this last case, and since each element in $B$ is adjacent to $v_{i+1}$, the edge $N[v_{i+1}]$ may not be the private edge of such an $x$ in $Z$.
        Hence $Z\cap X\cap R=\emptyset$, that is, $Z\cap R\subseteq Y$.
        By \autoref{prop:red-part} there is exactly one vertex, call it $y$, belonging to $Z\cap Y$.
        Finally, note that in addition to $N[v_{i+1}]$, each \(b \in B\) hits a single edge \(N[b] \in \Delta_{i+1}\) that only $b$ can hit in $X$.
        Hence \(B \setminus N(r) \subseteq Z\) and \hyperref[it:blue-red]{Item \ref*{it:blue-red}} follows.
    \end{proof}
\end{lemma}

Note that the conditions of the items in \autoref{lem:extension-characterization} can be checked in polynomial time, and provide an upper bound of $O(n^2)$ on the number of minimal transversals of $\Delta_{i+1}$.
Clearly these minimal transversals can be constructed within the same time and using polynomial space.
We formalize this as follows.

\begin{corollary}\label{cor:mds-extensions}
    The set $Tr(\Delta_{i+1})$ can be listed in polynomial time and space.
\end{corollary}

We conclude the following by combining 
\hyperref[cor:all-extensions]{Corollaries \ref*{cor:all-extensions}} and~\ref{cor:mds-extensions}.
This answers positively \autoref{qu:mds}.

\begin{theorem}\label{thm:mds-ch-bip}
    There is a polynomial delay and space algorithm enumerating minimal dominating sets in chordal bipartite graphs.
\end{theorem}

Let us now argue that with small modifications we can get a polynomial-delay algorithm for \tdomenum{} on the same class.
Note that the minimal extension problem is almost identical, with the caveat that the hyperedges we must hit are now of the form $N(x)$ for $x=v_{i+1}$ or $x\in X$, i.e., we are now dealing with open instead of closed neighborhoods.
A similar analysis as conducted in the proof of \autoref{lem:extension-characterization} shows that minimal extensions are red subsets dominating blue elements plus $v_{i+1}$ if $N(v_{i+1})\subseteq V_{i+1}$.
We note that this becomes an instance of proper red-blue domination (where only red elements may be picked) which is precisely what is solved in \cite{golovach2016enumerating} for chordal bipartite graphs.
However in our case, due to the instance being a bipartite chain graph, the problem is trivial.
Indeed, it may be seen that solutions have size at most two, intersecting $X$ and $Y$ on at most one vertex.
As such, we have an alternative proof of \autoref{thm:mtds-ch-bip}, which was already known by \cite{golovach2016enumerating}.

\thmtds*

\section{Minimal connected dominating sets}\label{sec:mcds}

Recall that \cdomenum{} amounts to list the minimal transversals of the (implicit) hypergraph of minimal separators of a given graph $G$, i.e., \autoref{prop:mcds-separators}.
By \autoref{prop:chbip-sep-enum}, this hypergraph can be constructed in polynomial time if $G$ is chordal bipartite.
In this section, we moreover show that this hypergraph has bounded conformality in that case. 
As a corollary, we can use the algorithm of Khachiyan et al.~in~\cite{khachiyan2007conformality} to list minimal connected dominating sets in incremental-polynomial time in chordal bipartite graphs, providing a positive answer to \autoref{qu:mcds}.

Let us define the notion of conformality introduced by Berge in~\cite{berge1984hypergraphs}. 
We say that a hypergraph $\H$ is of \emph{conformality} $c$ if the following property holds for every subset $X\subseteq V(\H)$: $X$ is contained (as a subset) in a hyperedge of $\H$ whenever each subset of $X$ of cardinality at most $c$ is contained (as a subset) in a hyperedge of $\H$. 
See also~\cite{berge1984hypergraphs,khachiyan2007conformality} for other equivalent characterizations of bounded conformality.
We recall that a hypergraph $\H$ is called Sperner if $E_1\not\subseteq E_2$ for any two distinct hyperedges $E_1,E_2$ in $\H$.

Khachiyan, Boros, Elbassioni, and Gurvich proved the following.

\begin{theorem}[\cite{khachiyan2007conformality}]\label{theorem:conformality}
    The minimal transversals of a Sperner hypergraph of bounded conformality can be enumerated in incre\-mental-poly\-nomial time.
\end{theorem}

\iflongelse{
We show that the family of minimal separators has bounded conformality in chordal bipartite graphs.}{Using the properties of minimal separators of chordal bipartite graphs given in \autoref{sec:prelim}, we obtain the following, whose proof is moved to appendix, and derive the desired algorithm as a corollary of \autoref{theorem:conformality} and \autoref{lemma:conf}.}

\begin{lemma}\label{lemma:conf}
    If $G$ is chordal bipartite then $\S(G)$ has conformality $5$.
\end{lemma}

\iflongelse{
\begin{proof}
    Let us put $\H\coloneq \S(G)$.
    Assume that $\H$ is not of conformality ${5}$, i.e., there is a subset $X\subseteq V(\H)$ that is not contained in a hyperedge of $\H$, but each of its subsets $Y\subseteq X$ of size at most $5$ are contained in a hyperedge of $\H$.
    Here, by \emph{contained} we mean as a subset.
    We consider $X=\{x_1,\dots, x_p\}$ of minimum cardinality.
    Then ${p\geq 6}$ as otherwise, by assumption, $X$ would be contained in a hyperedge of $\H$.
    Moreover, for every $x_i\in X$, note that there must exist a hyperedge $E_i$ of $\H$ such that $E_i\cap X=X\setminus \{x_i\}$: we prove this depending on the value of $p$.
    In the first case, if $p \geq 7$, the nonexistence of such $E_i$ implies that $X'=X\setminus \{x_i\}$ is not contained in a hyperedge of $\H$, and still every subset $Y\subseteq X'$ of size at most $5$ is, contradicting the minimality of $X$; if, however, $p=6$, by definition we have that every $Y \subsetneq X$ of size 5 is contained in some $E_Y \in E(\H)$, but $X \nsubseteq E_Y$ since $\H$ is not of conformality 5.

    Since $X$ has cardinality at least $6$, there exist at least three vertices $a_1,a_2,a_3$ in $X$ lying in one side of the bipartition of $G$, call it $A$.
    Let $E_1$ be an edge of $\H$, i.e., a minimal separator of $G$ that does not contain $a_1$, which exists by our previous arguments; moreover, $E_1 \cap X = X \setminus \{a_1\}$, implying $a_2,a_3 \in E_1$.
    By the minimality of $E_1$ and \autoref{prop:close} there are at least two components that are close to $E_1$.
    Let $C_1$ be one that does \textit{not} contain $a_1$.
    By \autoref{prop:close-neighbor} there exists $c_1$ in $C_1$ that satisfies $N(c_1) \cap A = E_1 \cap A$, and so $c_1$ is adjacent to both $a_2$ and $a_3$.
    However note that $c_1$ is not adjacent to $a_1$ as we chose $C_1$ to be the component not containing $a_1$.
    
    Consider now the separator $E_2$.
    Again, by minimality there are at least two components that are close to $E_2$. 
    Let $C_2$ be the one that does not contain $a_2$.
    By \autoref{prop:close-neighbor} there exists $c_2$ in $C_2$ that is adjacent to both $a_1$ and $a_3$.
    However $a_2$ and $c_2$ are non adjacent.
    
    We repeat this procedure for $E_3$ and obtain $c_3$ that is adjacent to both $a_1$ and $a_2$ but not $a_3$.
    Note that as $a_1,a_2,a_3$ lie in $A$, $c_1,c_2,c_3$ lie in the other side of the bipartition and are thus pairwise non adjacent.
    We obtain an induced cycle $(a_1, c_1, a_2, c_2, a_3, c_3, a_1)$ contradicting the fact that $G$ is $C_{\geq 6}$-free.
\end{proof}}{}

\iflongelse{We derive the following as a corollary of \autoref{theorem:conformality} and \autoref{lemma:conf}.}{\vspace{-.2cm}}

\begin{theorem}\label{thm:mcds-incp}
    The minimal connected dominating sets of a chordal bipartite graph can be listed in incremental polynomial time.
\end{theorem}

\section{Limitations of the sequential method}\label{sec:sequential-limit}

It is already known that the sequential method may not directly provide output-polynomial time algorithms for \transenum{} in general; see e.g.~\cite{bartier2024hypergraph} for explicit statements.
We adapt the construction in~\cite{bartier2024hypergraph} to show that this still holds for hypergraphs of (open or closed) neighborhoods of bipartite graphs, i.e., that we may not expect to get tractable algorithms for \tdomenum{} and \domenum{} using this technique.

We begin with \domenum{} and later detail the modifications to be performed for the statement to hold for \tdomenum{}.
In the following by $\NN(G)$ we mean the hypergraph of closed neighborhoods of $G$, whose minimal transversals are precisely the minimal dominating sets of $G$; see \autoref{sec:prelim}.

Recall that in \mcis{} problem, we are given a graph $G$, an integer $k$, and a
partition $\mathcal{V}=(V_1, V_2, \dots, V_k)$ of $V(G)$ such that $V_i$ is an independent set for every $1\leq i \leq k$.
Sets $V_1,\dots,V_k$ are usually referred to as color classes.
The goal is then to decide whether there exists an independent set in $G$ containing precisely one vertex from each $V_i$, i.e., that is multicolored.
See~\cite{cygan2015parameterized} for more details on this problem.

\begin{figure}
    \centering
    \includegraphics[scale=\figscale]{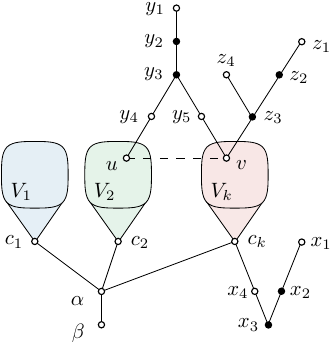}
\caption{The construction of $H$, where vertices are ordered arbitrarily except for $\alpha$ and $\beta$ which are put last, in this order. Color classes $V_1,\dots,V_k$ induce independent sets in $H$. The dashed edge denote an edge of $G$ that is not part of $H$. Black vertices represent $T^\star$. Note that there is one copy of the depicted gadgets for each color class, vertex, and edge of $G$.}\label{fig:example-reduction}
\end{figure}

\iflongelse{}{Due to space constraints, we omit the formal description of the reduction and only include the  statement we derive. We refer the reader to \autoref{fig:example-reduction} for an idea of the construction, and to the appendix for its formal description and proof. The idea is to show that a minimal transversal $T^\star$ of $\H$ induced by all but two of its vertices can be extended into a non-trivial child if and only if an instance $G$ of multicolored independent set admits a solution.}

\begin{lemma}\label{lem:hardness-sequential-bipartite}
    Let \((G, \mathcal{V})\) be an instance of \mcis{}. Then there exists a bipartite graph \(H\), an ordering \(v_{1}, \dots, v_{n}\) of its vertices, and a minimal transversal $T^\star$ of $\NN(H)_{n-2}$, such that the sets in 
    $\children(T^\star, n-2)$ are all multicolored independent sets of \(G\) except for one.
\end{lemma}

\iflongelse{
\begin{proof}
    Let \((G, \mathcal{V})\) be an instance of \mcis{} where \(\mathcal{V} = (V_{1}, \dots, V_{k})\) is a partition of $G$ and each $V_i$ is one of $k$ color classes, $n = |V(G)|$ and $m = |E(G)|$.
    We describe the construction of a bipartite graph $H$ from $G$ as follows, and create a set $T^\star$ of vertices along the way.
    See \autoref{fig:example-reduction} for an illustration.

    We start by adding to $H$ the set of vertices $V(G)$, but none of the edges of $G$.
    For each color class $V_{i}$, $1\leq i\leq k$ we create a special vertex $c_i$ that we make adjacent to all vertices of $V_{i}$.
    Then, again for each color class, we create a path on four vertices $x_1,x_2,x_3,x_4$ and attach it to $c_i$ by adding the edge $c_ix_4$.
    We add $x_2$ and $x_3$ to $T^\star$.
    The role of this gadget is to leave $c_j$ non-dominated by $T^\star$, yet to forbid $c_j$ in a minimal dominating set containing $T^\star$.
    In the following, we define $C = \{c_i : i \in \{1,\dots, k\}\}$.
    
    For every edge $uv$ in $G$ we construct a subdivided claw $y_1, y_2, y_3, y_4, y_5, u, v$, with $N(y_2) = \{y_1, y_3\}$, $N(y_4) = \{y_3, u\}$, $N(y_5) = \{y_3, v\}$ and no other edges touch the $y$ vertices.
    Also, we add $y_2$ and $y_3$ to $T^\star$.
    The objective of this gadget is to forbid the selection of \emph{both} endpoints \(u\) and \(v\) of an edge in a minimal dominating set containing $T^\star$.
    Note that we construct one such gadget per edge of $G$, but that the vertices $u$ and $v$ are shared by other edge gadgets.
    
    For every vertex $v$ of $G$, we create a distinct path on four vertices $z_1,z_2,z_3,z_4$, and attach it to $v$ with the edge $vz_3$.
    We add $z_2$ and $z_3$ to $T^\star$.
    This ensures that each such $v$ is dominated by $T^\star$.
    Moreover, note that picking $v$ in a minimal dominating set  containing $T^\star$ will not steal private neighbors of an element in $T^\star$, and will dominate $c_i$ for $i$ such that $v\in V_i$.

    Finally, we create two additional vertices $\alpha$ and $\beta$ with $\alpha$ adjacent to each $c_i$, $1\leq i\leq k$, and $\beta$ only adjacent to $\alpha$. 
    We fix the ordering of vertices of $H$ to be arbitrary but ending with $\alpha,\beta$ in this order.
    
    This completes the construction of $H$.
    It has $O(n+m)$ vertices and edges.
    Moreover, it is easily seen that $H$ is bipartite, by taking $V(G)$ on one side of the bipartition, and noting that except for $\alpha$ and $\beta$ the gadgets only create even paths between pairs of vertices attached to $V(G)$. 
    Then $\alpha$ is made adjacent to the $c_j$'s which lie in the opposite partition of $V(G)$, and $\beta$ is a pendant vertex.
    Thus they do not create odd cycles.
    See also \autoref{fig:example-reduction} for a representation of $H$ in layers where each edge is between consecutive layers. 

    Let us define $\H\coloneq \NN(H)$ and show that $T^\star$ is a minimal transversal of $\H_{n-2} = \H\left[V(H) \setminus \{\alpha, \beta\}\right]$.
    First note that all neighborhoods except for those of $c_i$'s, $\alpha$ and $\beta$ are included in the $n-2$ first vertices, hence belong to $\H_{n-2}$.
    Moreover, all these neighborhoods are intersected by $T^\star$ as every gadget is dominated by $T^\star$ by construction, and each vertex $v\in V(G)$ is dominated by the vertex $z_3$ of its corresponding vertex gadget.
    Furthermore, $T^\star$ is minimal with that property as each of its element $t$ has a private neighbor in its corresponding gadget, whose edge is thus only intersected by $t$ in $T^\star$.

    Let us now consider the different possibilities for a minimal transversal $T$ of $\H_{n-1}$ to contain $T^\star$ as a subset.
    Note that only the neighborhoods of $c_i$ for each $1\leq i\leq k$ are to be dominated. 
    To do so, we either pick $\alpha$, or must select a subset $X$ of $V(G)\cup C$ in order to minimally hit $N[c_i]$ for all $1\leq i\leq k$.
    Note however that as $x_4$ is the only private neighbor of $x_3$ in the color class gadget of $V_i$, $X$ cannot contain $c_i$ nor $x_1$ for any color.
    Hence $X\subseteq V(G)$.
    Note also that no two distinct elements of $X$ intersect a same color class $V_i$ as their only private neighbor would be $c_i$, and hence one of the two would be redundant.
    Moreover by the edge gadget, since $y_4$ and $y_5$ are the only two private neighbors of $y_3$, $X$ cannot contain both $u$ and $v$ for any edge $uv$ of $G$.
    Thus $X$ intersects each color class precisely once, and does not contain two vertices $u,v$ with $uv\in E(G)$.
    We conclude that $X$ is a multicolored independent set of $G$, hence that $T^\star$ can be extended into a minimal transversal of $\H_{n-1}$ using a different set than $\{\alpha\}$ if and only if $G$ contains a multicolored independent set. 
    
    Finally, recall that by 
    \autoref{lem:at-least-one-child} and
    \autoref{lem:extensions} either $\children(T^\star,n-2)=\{T^\star\}$, or sets in $\children(T^\star,n-2)$ are of the form $T^\star\cup X$ for some $X\in Tr(\Delta_{n-1})$ where 
    \[
        \Delta_{n-1}\coloneq \{E\in \H_{n-1} : E\cap T^{\star}=\emptyset\}=\{N[c_i] : 1\leq i\leq k\}.
    \]
    Since $T^\star$ is not a transversal of $\H_{n-1}$ we are in the second case.
    Moreover, by the above discussion, each such $X\in Tr(\Delta_{n-1})$ is either $\{\alpha\}$ or is a multicolored independent set of $G$.
    We thus need to show that any multicolored independent set of $G$ yields a children to conclude the proof.

    Let $X$ be a multicolored independent set of $G$ and suppose that $T^\star\cup X$ is not a child of $(T^\star,n-2)$.
    Then during the computation of the parent of $T^\star\cup X$, some element $v\in T^\star$ is removed.
    Hence such a $v$ belongs to one of the edge, vertex, of color class gadgets.
    But for these element to lose a private neighbor it should be that $X$ contains two elements of a same color class, or two endpoints of an edge, a contradiction.
\end{proof}}{}

A consequence of \autoref{lem:hardness-sequential-bipartite} is that the sequential approach may not be directly applied to solve~\domenum{}, even in bipartite graphs.
Moreover, \iflongelse{note that }{}in our construction, no vertex $v$ of $T^\star$ is made self-private, i.e., no $v$ has its closed neighborhood as a private edge.
Furthermore, the elements in the color classes are dominated by $T^\star$, and their role is thus limited to dominate the $c_i$'s in an extension. 
Hence, by keeping the same construction, all the arguments of the \iflongelse{above}{} proof hold for the hypergraph of \emph{open} neighborhoods of $G$.
We derive that the sequential approach may not be directly be used for \tdomenum{} neither%
\iflongelse{}{.}%
\iflongelse{, as summarized in \autoref{lem:hardness-sequential-bipartite-mtds}, where $\NN^o(H)$ is the hypergraph of open neighborhoods of $H$.

\begin{lemma}\label{lem:hardness-sequential-bipartite-mtds}
    Let \((G, \mathcal{V})\) be an instance of \mcis{}. Then there exists a bipartite graph \(H\), an ordering \(v_{1}, \dots, v_{n}\) of its vertices, and a minimal transversal $T^\star$ of $\NN^o(H)_{n-2}$, such that the sets in 
    $\children(T^\star, n-2)$ are all multicolored independent sets of \(G\) except for one. 
\end{lemma}}{}

\section{\transenum{} and bipartite graphs\protect\footnote{After the submission of this work, it has come to our attention that there is a recent stronger result on the subject~\cite{capacitated_vc_enum}. We nevertheless leave the proof as is since it is simple and for self-containment.
}}
\label{sec:mcds-hardness}

Our final result is a proof that \transenum{} reduces to \cdomenum{} on bipartite graphs, which we are not aware of being previously known or directly implied by other results. In fact, by \autoref{prop:mcds-separators}, if we can construct a graph where (almost) every minimal separator is in a one to one correspondence to our hyperedges, we are done.

Let $\H$ be the input hypergraph to \transenum{} with $E(\H) = \{E_1, \dots, E_m\}$ and $V(\H) = \{u_1, \dots, u_n\}$ and assume w.l.o.g.~that $\H$ is a Sperner hypergraph. To obtain $G$, the input to \cdomenum{}, we start with the incidence bipartite graph of $\H$, that is, $\{v_1, \dots, v_n\}$, $\{e_1, \dots, e_m\} \subset V(G)$, where $v_i$ corresponds to $u_i \in V(\H)$ and $e_j$ corresponds to $E_j \in E(\H)$, and $v_ie_j \in E(G)$ if and only if $u_i \in E_j$. To conclude the construction of $G$, we add two vertices $v'$ and $v^\star$, the edge $v'v^\star$, and the edges $v^\star v_i$ for all $1\leq i\leq n$.
\iflongelse{}{Due to space constraints the proofs of the following are omitted as well.}

\begin{lemma}\label{lem:mcds-hardness}
    The minimal separators of $G$ are $N(e_1),\dots, N(e_m)$ and $\{v^\star\}$.
\end{lemma}

\iflongelse{
\begin{proof}
    To see that $\{v^\star\} \in \S(G)$, observe that $G \setminus \{v^\star\}$ has two connected components: $\{v'\}$ and $V(G) \setminus \{v', v^\star\}$.
    As such, no other $S \in \S(G)$ contains $v^\star$, so it holds that every $v_i \notin S$ is in the same connected component as $v^\star$ and $v'$.
    Let $E_j \in E(\H)$ and observe that $N(e_j)$ separates $e_j$ from the rest of $G \setminus N(e_j)$; moreover, note that, since $\H$ is a Sperner hypergraph, every other $e_\ell \in V(G)$ has a neighbor $v_k \notin N(e_j)$ (as no two hyperedges satisfy a containment relation), so it holds that $G \setminus N(e_j)$ has exactly two components: $\{e_j\}$ and the rest of the graph.
    Thus, if we add back one vertex from $N(e_j)$, $e_j$ is no longer disconnected as it creates a $e_j$--$v^\star$ path.
    Now, towards a contradiction, assume that there is another $S \in \S(G)$ that is different from every $N(e_j)$ and from $\{v^\star\}$.
    Let $\{v', v^\star, v_1, \dots, v_n\} \setminus S \subseteq C_1$ be a non-empty connected component of $G \setminus S$; note that, because $V(G) \setminus (S \cup C_1) \subseteq \{e_1, \dots, e_m\}$, the other connected components of $G \setminus S$ (which must exist as $S$ is a separator) must have size one, but this implies that each $e_j \notin C_1 \cup S$ has $N(e_j) = N(e_j) \subseteq S$, a contradiction.
\end{proof}}{\vspace{-.3cm}}

\begin{theorem}\label{thm:mcds-hardness}
    \cdomenum{} on bipartite graphs is at least as hard as \transenum{}.
\end{theorem}

\iflongelse{
\begin{proof}
    By our construction and \autoref{lem:mcds-hardness}, each $T_G \in Tr(\S(G))$ contains $v^\star$ and $T_G \setminus \{v^\star\}$ minimally hits each $N(e_j) = \{v_i : u_i \in E_j\}$; thus $T_\H(T_G) = \{u_i : v_i \in T_G \setminus \{v^\star\} \}$ is a bijection between $Tr(\S(G))$ and $Tr(\H)$.
    By \autoref{prop:mcds-separators}, $\mcds(G)=Tr(\S(G))$, so there is a bijection between $\mcds(G)$ and $Tr(\H)$ and the result follows.
\end{proof}}{\vspace{-.3cm}}

\section{Discussion}\label{sec:conclusion}

In this paper, we investigated the complexity of \domenum{} and its two variants \tdomenum{} and \cdomenum{} in chordal bipartite graphs.
We obtained polynomial-delay algorithms for the first two problems, and an incremental-polynomial time algorithm for the last problem.
Then we gave evidence that the techniques used for \domenum{} and \tdomenum{} may not be pushed directly to bipartite graphs.
As of \cdomenum{}, we proved a stronger statement that \cdomenum{} is as hard as \transenum{} in bipartite graphs, with the later problem being open since more than forty years~\cite{eiter1995identifying}.

We leave open the existence of polynomial-delay algorithms for \domenum{} and \tdomenum{} in bipartite graphs, and for \cdomenum{} in chordal bipartite graphs.
Concerning this last case, we note that obtaining polynomial delay is open for \transenum{} on hypergraph of bounded dimension, which define a proper subclass of hypergraphs of bounded conformality.
Hence, in order to improve \autoref{thm:mcds-incp}, one either has to find another strategy, or should aim at improving the results in~\cite{eiter1995identifying,khachiyan2007conformality} first.

%% file: main.bbl
\newcommand{\etalchar}[1]{$^{#1}$}
\begin{thebibliography}{AKFG{\etalchar{+}}22}

\bibitem[AKFG{\etalchar{+}}22]{abu2022connected}
Faisal~N. Abu-Khzam, Henning Fernau, Benjamin Gras, Mathieu Liedloff, and Kevin
  Mann.
\newblock Enumerating minimal connected dominating sets.
\newblock In {\em 30th Annual European Symposium on Algorithms (ESA 2022)}.
  Schloss Dagstuhl-Leibniz-Zentrum f{\"u}r Informatik, 2022.

\bibitem[BBC00]{berry2000generating}
Anne Berry, Jean-Paul Bordat, and Olivier Cogis.
\newblock Generating all the minimal separators of a graph.
\newblock {\em International Journal of Foundations of Computer Science},
  11(03):397--403, 2000.

\bibitem[BBHK15]{lorentz2015open}
Hans Bodlaender, Endre Boros, Pinar Heggernes, and Dieter Kratsch.
\newblock {\em Open Problems of the Lorentz Workshop, "Enumeration Algorithms
  using Structure"}.
\newblock Department of Information and Computing Sciences Utrecht University,
  Utrecht, The Netherlands, 2015.

\bibitem[BDH{\etalchar{+}}20]{bonamy2020kt}
Marthe Bonamy, Oscar Defrain, Marc Heinrich, Micha{\l} Pilipczuk, and
  Jean-Florent Raymond.
\newblock Enumerating minimal dominating sets in {$K_t$}-free graphs and
  variants.
\newblock {\em ACM Transactions on Algorithms (TALG)}, 16(3):1--23, 2020.

\bibitem[BDK{\etalchar{+}}24]{brosse2024hardness}
Caroline Brosse, Oscar Defrain, Kazuhiro Kurita, Vincent Limouzy, Takeaki Uno,
  and Kunihiro Wasa.
\newblock On the hardness of inclusion-wise minimal separators enumeration.
\newblock {\em Information Processing Letters}, 185:106469, 2024.

\bibitem[BDM24]{bartier2024hypergraph}
Valentin Bartier, Oscar Defrain, and Fionn {Mc~Inerney}.
\newblock Hypergraph dualization with {FPT}-delay parameterized by the
  degeneracy and dimension.
\newblock In {\em International Workshop on Combinatorial Algorithms}, pages
  111--125. Springer, 2024.

\bibitem[BDMN20]{bonamy2020comp}
Marthe Bonamy, Oscar Defrain, Piotr Micek, and Lhouari Nourine.
\newblock Enumerating minimal dominating sets in the (in)comparability graphs
  of bounded dimension posets.
\newblock {\em arXiv preprint arXiv:2004.07214}, 2020.

\bibitem[Ber84]{berge1984hypergraphs}
Claude Berge.
\newblock {\em Hypergraphs: combinatorics of finite sets}, volume~45.
\newblock Elsevier, 1984.

\bibitem[BFL{\etalchar{+}}22]{blasius2022efficiently}
Thomas Bl{\"a}sius, Tobias Friedrich, Julius Lischeid, Kitty Meeks, and Martin
  Schirneck.
\newblock Efficiently enumerating hitting sets of hypergraphs arising in data
  profiling.
\newblock {\em Journal of Computer and System Sciences}, 124:192--213, 2022.

\bibitem[CFK{\etalchar{+}}15]{cygan2015parameterized}
Marek Cygan, Fedor~V. Fomin, {\L}ukasz Kowalik, Daniel Lokshtanov, D{\'a}niel
  Marx, Marcin Pilipczuk, Micha{\l} Pilipczuk, and Saket Saurabh.
\newblock {\em Parameterized algorithms}.
\newblock Springer, 2015.

\bibitem[CKMU19]{conte2019irredundant}
Alessio Conte, Mamadou~M. Kant{\'e}, Andrea Marino, and Takeaki Uno.
\newblock Maximal irredundant set enumeration in bounded-degeneracy and
  bounded-degree hypergraphs.
\newblock In {\em International Workshop on Combinatorial Algorithms}, pages
  148--159. Springer, 2019.

\bibitem[Cou09]{courcelle2009linear}
Bruno Courcelle.
\newblock Linear delay enumeration and monadic second-order logic.
\newblock {\em Discrete Applied Mathematics}, 157(12):2675--2700, 2009.

\bibitem[Die12]{diestel2012graph}
Reinhard Diestel.
\newblock {\em Graph Theory}.
\newblock Springer Berlin, Heidelberg, 2012.

\bibitem[DN19]{defrain2019neighborhood}
Oscar Defrain and Lhouari Nourine.
\newblock Neighborhood inclusions for minimal dominating sets enumeration:
  Linear and polynomial delay algorithms in {$P_7$}-free and {$P_8$}-free
  chordal graphs.
\newblock In {\em 30th International Symposium on Algorithms and Computation
  (ISAAC 2019)}. Schloss Dagstuhl-Leibniz-Zentrum fuer Informatik, 2019.

\bibitem[DNV21]{defrain2021translating}
Oscar Defrain, Lhouari Nourine, and Simon Vilmin.
\newblock Translating between the representations of a ranked convex geometry.
\newblock {\em Discrete Mathematics}, 344(7):112399, 2021.

\bibitem[EG95]{eiter1995identifying}
Thomas Eiter and Georg Gottlob.
\newblock Identifying the minimal transversals of a hypergraph and related
  problems.
\newblock {\em SIAM Journal on Computing}, 24(6):1278--1304, 1995.

\bibitem[EGM03]{eiter2003new}
Thomas Eiter, Georg Gottlob, and Kazuhisa Makino.
\newblock New results on monotone dualization and generating hypergraph
  transversals.
\newblock {\em SIAM Journal on Computing}, 32(2):514--537, 2003.

\bibitem[EHR08]{elbassioni2008some}
Khaled Elbassioni, Matthias Hagen, and Imran Rauf.
\newblock Some fixed-parameter tractable classes of hypergraph duality and
  related problems.
\newblock In {\em Parameterized and Exact Computation: Third International
  Workshop, IWPEC 2008, Victoria, Canada, May 14-16, 2008. Proceedings 3},
  pages 91--102. Springer, 2008.

\bibitem[EMG08]{eiter2008computational}
Thomas Eiter, Kazuhisa Makino, and Georg Gottlob.
\newblock Computational aspects of monotone dualization: A brief survey.
\newblock {\em Discrete Applied Mathematics}, 156(11):2035--2049, 2008.

\bibitem[ERR19]{elbassioni2019global}
Khaled Elbassioni, Imran Rauf, and Saurabh Ray.
\newblock A global parallel algorithm for enumerating minimal transversals of
  geometric hypergraphs.
\newblock {\em Theoretical Computer Science}, 767:26--33, 2019.

\bibitem[FK96]{fredman1996complexity}
Michael~L. Fredman and Leonid Khachiyan.
\newblock On the complexity of dualization of monotone disjunctive normal
  forms.
\newblock {\em Journal of Algorithms}, 21(3):618--628, 1996.

\bibitem[GDVL17]{gainer2017minimal}
Andrew Gainer-Dewar and Paola Vera-Licona.
\newblock The minimal hitting set generation problem: algorithms and
  computation.
\newblock {\em SIAM Journal on Discrete Mathematics}, 31(1):63--100, 2017.

\bibitem[GG78]{golumbic1978perfect}
Martin~Charles Golumbic and Clinton~F. Goss.
\newblock Perfect elimination and chordal bipartite graphs.
\newblock {\em Journal of Graph Theory}, 2(2):155--163, 1978.

\bibitem[GHK{\etalchar{+}}16a]{golovach2016enumerating}
Petr~A. Golovach, Pinar Heggernes, Mamadou~M. Kant{\'e}, Dieter Kratsch, and
  Yngve Villanger.
\newblock Enumerating minimal dominating sets in chordal bipartite graphs.
\newblock {\em Discrete Applied Mathematics}, 199:30--36, 2016.

\bibitem[GHK16b]{golovach2016connected}
Petr~A. Golovach, Pinar Heggernes, and Dieter Kratsch.
\newblock Enumerating minimal connected dominating sets in graphs of bounded
  chordality.
\newblock {\em Theoretical Computer Science}, 630:63--75, 2016.

\bibitem[GHK{\etalchar{+}}18]{golovach2018lmimwidth}
Petr~A. Golovach, Pinar Heggernes, Mamadou~M. Kant{\'e}, Dieter Kratsch,
  Sigve~H. S{\ae}ther, and Yngve Villanger.
\newblock Output-polynomial enumeration on graphs of bounded (local) linear
  {MIM}-width.
\newblock {\em Algorithmica}, 80(2):714--741, 2018.

\bibitem[GHKS20]{golovach2020connected}
Petr~A. Golovach, Pinar Heggernes, Dieter Kratsch, and Reza Saei.
\newblock Enumeration of minimal connected dominating sets for chordal graphs.
\newblock {\em Discrete Applied Mathematics}, 278:3--11, 2020.

\bibitem[GHKV15]{golovach2015flipping}
Petr~A. Golovach, Pinar Heggernes, Dieter Kratsch, and Yngve Villanger.
\newblock An incremental polynomial time algorithm to enumerate all minimal
  edge dominating sets.
\newblock {\em Algorithmica}, 72(3):836--859, 2015.

\bibitem[GMKT97]{gunopulos1997data}
Dimitrios Gunopulos, Heikki Mannila, Roni Khardon, and Hannu Toivonen.
\newblock Data mining, hypergraph transversals, and machine learning.
\newblock In {\em PODS}, pages 209--216. ACM, 1997.

\bibitem[Gol04]{golumbic2004algorithmic}
Martin~Charles Golumbic.
\newblock {\em Algorithmic graph theory and perfect graphs}.
\newblock Elsevier, 2004.

\bibitem[JYP88]{johnson1988generating}
David~S. Johnson, Mihalis Yannakakis, and Christos~H. Papadimitriou.
\newblock On generating all maximal independent sets.
\newblock {\em Information Processing Letters}, 27(3):119--123, 1988.

\bibitem[KBEG07]{khachiyan2007conformality}
Leonid Khachiyan, Endre Boros, Khaled Elbassioni, and Vladimir Gurvich.
\newblock On the dualization of hypergraphs with bounded edge-intersections and
  other related classes of hypergraphs.
\newblock {\em Theoretical Computer Science}, 382(2):139--150, 2007.

\bibitem[KK94]{kloks1994finding}
Ton Kloks and Dieter Kratsch.
\newblock Finding all minimal separators of a graph.
\newblock In {\em STACS 94: 11th Annual Symposium on Theoretical Aspects of
  Computer Science Caen, France, February 24--26, 1994 Proceedings 11}, pages
  759--768. Springer, 1994.

\bibitem[KK98]{kloks1998listing}
Ton Kloks and Dieter Kratsch.
\newblock Listing all minimal separators of a graph.
\newblock {\em SIAM Journal on Computing}, 27(3):605--613, 1998.

\bibitem[KKMO24]{capacitated_vc_enum}
Yasuaki Kobayashi, Kazuhiro Kurita, Yasuko Matsui, and Hirotaka Ono.
\newblock Enumerating minimal vertex covers and dominating sets with capacity
  and/or connectivity constraints.
\newblock In Adele~Anna Rescigno and Ugo Vaccaro, editors, {\em Combinatorial
  Algorithms}, pages 232--246, Cham, 2024. Springer Nature Switzerland.

\bibitem[KLM{\etalchar{+}}15a]{kante2015chordal}
Mamadou~M. Kant{\'e}, Vincent Limouzy, Arnaud Mary, Lhouari Nourine, and
  Takeaki Uno.
\newblock A polynomial delay algorithm for enumerating minimal dominating sets
  in chordal graphs.
\newblock In {\em International Workshop on Graph-Theoretic Concepts in
  Computer Science}, pages 138--153. Springer, 2015.

\bibitem[KLM{\etalchar{+}}15b]{kante2015line}
Mamadou~M. Kant{\'e}, Vincent Limouzy, Arnaud Mary, Lhouari Nourine, and
  Takeaki Uno.
\newblock Polynomial delay algorithm for listing minimal edge dominating sets
  in graphs.
\newblock In {\em Workshop on Algorithms and Data Structures}, pages 446--457.
  Springer, 2015.

\bibitem[KLMN12]{kante2012neighbourhood}
Mamadou~M. Kant{\'e}, Vincent Limouzy, Arnaud Mary, and Lhouari Nourine.
\newblock On the neighbourhood helly of some graph classes and applications to
  the enumeration of minimal dominating sets.
\newblock In {\em International Symposium on Algorithms and Computation}, pages
  289--298. Springer, 2012.

\bibitem[KLMN14]{kante2014split}
Mamadou~M. {Kant\'e}, Vincent Limouzy, Arnaud Mary, and Lhouari Nourine.
\newblock On the enumeration of minimal dominating sets and related notions.
\newblock {\em SIAM Journal on Discrete Mathematics}, 28(4):1916--1929, 2014.

\bibitem[KLP11]{kloks2011feedback}
Ton Kloks, Ching-Hao Liu, and Sheung-Hung Poon.
\newblock Feedback vertex set on chordal bipartite graphs.
\newblock {\em arXiv preprint arXiv:1104.3915}, 2011.

\bibitem[KN14]{kante2014encyclopedia}
Mamadou~M. Kant{\'e} and Lhouari Nourine.
\newblock Minimal dominating set enumeration.
\newblock In Ming-Yang Kao, editor, {\em Encyclopedia of Algorithms}, pages
  1--5. Springer US, Boston, MA, 2014.

\bibitem[KPS93]{kavvadias1993horn}
Dimitris Kavvadias, Christos~H. Papadimitriou, and Martha Sideri.
\newblock On {Horn} envelopes and hypergraph transversals.
\newblock In {\em International Symposium on Algorithms and Computation}, pages
  399--405. Springer, 1993.

\bibitem[KWUA19]{kurita2019efficient}
Kazuhiro Kurita, Kunihiro Wasa, Takeaki Uno, and Hiroki Arimura.
\newblock An efficient algorithm for enumerating chordal bipartite induced
  subgraphs in sparse graphs.
\newblock In {\em International Workshop on Combinatorial Algorithms}, pages
  339--351. Springer, 2019.

\bibitem[SL97]{shen1997efficient}
Hong Shen and Weifa Liang.
\newblock Efficient enumeration of all minimal separators in a graph.
\newblock {\em Theoretical computer science}, 180(1-2):169--180, 1997.

\bibitem[TIAS77]{tsukiyama1977new}
Shuji Tsukiyama, Mikio Ide, Hiromu Ariyoshi, and Isao Shirakawa.
\newblock A new algorithm for generating all the maximal independent sets.
\newblock {\em SIAM Journal on Computing}, 6(3):505--517, 1977.

\bibitem[Ueh02]{uehara2002linear}
Ryuhei Uehara.
\newblock Linear time algorithms on chordal bipartite and strongly chordal
  graphs.
\newblock In {\em Automata, Languages and Programming: 29th International
  Colloquium, ICALP 2002 M{\'a}laga, Spain, July 8--13, 2002 Proceedings 29},
  pages 993--1004. Springer, 2002.

\end{thebibliography}
